\documentclass{ifacconf}

\usepackage[utf8]{inputenc}
\usepackage{natbib}
\usepackage{amsmath,amssymb,amsfonts}
\usepackage{mathtools}
\usepackage{enumerate}
\usepackage{algorithmic}
\usepackage{graphicx}
\usepackage{textcomp}
\usepackage{xcolor}
\usepackage{multirow}
\usepackage{cuted}
\usepackage{xstring}
\usepackage{microtype}
\usepackage{enumitem}

\graphicspath{{./fig/} {./programs/} {../fig/}}

\providecommand{\proofname}{Proof}
\newcommand\qedsymbol{\hfill$\blacksquare$}

\makeatletter
\renewenvironment{pf}[1][\proofname]{\par
  \normalfont \topsep6\p@\@plus6\p@\relax
  \trivlist
  \item[\hskip\labelsep
        \itshape
    #1\@addpunct{.}]\ignorespaces
}{%
  \nolinebreak\qedsymbol\endtrivlist\@endpefalse
}
\makeatother

\makeatletter
\newcommand{\pushright}[1]{\ifmeasuring@#1\else\omit\hfill$\displaystyle{#1}$\fi\ignorespaces}
\newcommand{\pushleft}[1]{\ifmeasuring@#1\else\omit$\displaystyle{#1}$\hfill\fi\ignorespaces}
\makeatother

\newcommand{\emp}[1]{\ensuremath{\text{EMP}_{#1}}}

\definecolor{magenta}{cmyk}{0.1, 1, 0, 0}
\definecolor{greenedu}{cmyk}{1, 0, 1, 0.1}
\definecolor{cyanedu}{cmyk}{1, 0, 0, 0.1}
\definecolor{brown}{cmyk}{0.2, 0.63, 0.84, 0.2}
\definecolor{BZcolor}{cmyk}{0, 0.35, 1, 0}

\definecolor{DOR}{RGB}{255, 140, 0}
\definecolor{SOR}{RGB}{255, 79, 0}
\definecolor{EOR}{RGB}{186, 22, 12}

\providecommand{\com[2]}{\begin{tt}[#1: #2]\end{tt}}
\definecolor{red}{rgb}{1,0,0}

\definecolor{blu}{rgb}{0,0,1}

\definecolor{gre}{rgb}{0,0.7,0.3}

\newtheorem{theorem}{Theorem}[section]
\newtheorem{proposition}{Proposition}[section]
\newtheorem{lemma}{Lemma}[section]

\newtheorem{definition}{Definition}[section]

\begin{document}
\begin{frontmatter}
	
	\title{Optimal excitation and measurement patterns for networks with tree topology}
	
	\thanks[footnoteinfo]{{This work was supported by Conselho Nacional de Desenvolvimento de Científico e Tecnológico (CNPq)} through personal grants to both authors.}
	
	\author[First]{Eduardo Mapurunga} 
	\author[Second]{Alexandre Sanfelici Bazanella} 
	
	\address[First]{Universidade Federal do Ceará (UFC) - Itapajé-CE, Brazil.
	email: mapurunga@ufc.br}
	\address[Second]{Data-Driven Control Group - Department of Automation and Energy, Universidade Federal do Rio Grande do Sul, Porto Alegre-RS, Brazil.
	email: bazanella@ufrgs.br}
	
	\begin{abstract}                
				  In this work we evaluate the excitation and measurement patterns (EMP) for networks with tree topology. 	
					We investigate guidelines for the selection of the minimal EMPs, i.e. those with the least number of excited and measured nodes combined, for which 
					the accuracy obtained, in terms of the trace of the asymptotic covariance matrix, is optimal.
					We introduce the concept of partial information matrix as a means to systematically obtain the information matrix for any dynamic network. 
					For a specific tree class, called cross, we show that the accuracy of a particular module depends on the magnitude of the parameters to be estimated. 
					Furthermore, when all factors are equal, it is best to excite. 
					We extend a topological condition for branches under which the accuracy of a particular module of the network is independent of the other parameters from the tree. 
					We provide a numerical analysis showing that our guidelines could be used as a selection tool for minimal EMPs for tree networks. 
	\end{abstract}
	
	
\end{frontmatter}

\section{Introduction}

Dynamic networks have been the focus of much research attention in the last decade
\citep{VanDenHof2013ComplexPem, BazanellaEtAlNetworkIdentificationPartial2019a, WeertsEtAlPredictionErrorIdentification2018,GeversEtAlPracticalMethodConsistent2018}.
A dynamic network can be represented by a digraph $\mathcal{G}$ whose nodes represent signals of the network and edges represent the dynamic relationship between the signals, usually represented by transfer functions.  
They provide a framework to analyze several scientific problems, such as opinion consensus in social networks, communication systems, regulatory gene networks, ecological system 
\citep{anderson-dabbene-proskurnikov-etal-dynamical-2020, bullo-lectures-2018, mesbahi2010graph, chang-dobbe-bhushan-etal-reconstruction-2016}, just to name a few. 

Network identification concerns obtaining transfer function models (modules) from network signals.
Prior to the identification of the models, one must guarantee that the modules of the network are uniquely identifiable. 
Several works have focused on the identifiability of the network modules.
In \cite{MapurungaBazanellaOptimalAllocationExcitation2021} it was introduced the concept of Excitation and Measurement Pattern (EMP), that is how the excitation signals and measurement should be allocated in a network, 
inspired by \cite{BazanellaEtAlNetworkIdentificationPartial2019a}. 
It turns out that if one wants to guarantee identifiability of a network, one could tackle the problem by designing EMPs for that purpose. 
An EMP is said to be {\em minimal} if it guarantees generic identifiability and it uses the minimal number of excitation and measurements combined. 
If one wants to go one step further, one could also think on how to select an EMP that achieves the most accurate estimates.  
This problem was originally formulated for state-space branches and cycles in \cite{MapurungaBazanellaOptimalAllocationExcitation2021}. 
In the case of branches, 
guidelines for the selection of the most accurate EMP was provided for branches with transfer functions as modules in \cite{mapurunga-bazanella-optimal-excitation-2023}.
In a nutshell, one must choose to excite nodes near the source and measure nodes near the sink, see Section \ref{sec:Dynetsetting} for formal definitions.
It was shown in \cite{MapurungaBazanellaOptimalAllocationExcitation2021,mapurunga-bazanella-optimal-excitation-2023} that the difference in accuracy between the most accurate EMP and the least accurate one is often of several orders of magnitude, which testifies to the relevance of the problem at hand.  

In this work, we tackle the same problem but for a more complex network topology, namely trees. 
We investigate what are the principles of designing minimal EMPs for trees such that the trace of the asymptotic covariance matrix is the least among all minimal EMPs. 
We start by providing the common framework for the analysis of trees in Section \ref{sec:Dynetsetting} and network identification in Section \ref{sec:Problem}. 
Our first contribution is to provide a systematic procedure to compute the information matrix for any dynamic network, this is the matter of Section \ref{sec:infmat}. 
We start the analysis by comparing whether we should excite or measure an internal node from a specific tree called cross and its dual network, this analysis is presented in Section \ref{sec:exciteormeasure}.
We extend a result about the accuracy of a particular module from \cite{mapurunga-bazanella-optimal-excitation-2023} to trees, provided in Section \ref{sec:justSNR}.
Based on the results obtained so far we provide a guideline for the selection of the most accurate minimal EMPs and present a numeric analysis to validate our insights in Section \ref{sec:21trees}.
Finally, we present our conclusions in Section \ref{sec:conclusions}.

\section{Dynamic Networks Setting}
\label{sec:Dynetsetting}

We consider dynamic networks composed of $n$ nodes which represent scalar internal signals.  
The nodes may have the influence of external signals 
and
the internal signals are accessible through noisy measurements. 
Such dynamic networks are  described by the network equations:
\begin{subequations}
	\begin{align}
	w(t) &= A^0w(t-1) + Br(t), \label{eq:S_internal}\\
	y(t) &= Cw(t) + e(t), \label{eq:S_external}
	\end{align}
\end{subequations}
where 
$y(t) \in \mathbb{R}^p$ is the network output,
$w(t) \in \mathbb{R}^n$, $r(t) \in \mathbb{R}^m $, $e(t) \in \mathbb{R}^p$ are the network's internal and external signals and the measurement noise, respectively.
The matrix $A^0 \in \mathbb{R}^{n \times n}$ is referred to as network matrix, whose entries are the parameters of the network. 
The forward shift operator is denoted as $q$, i.e. $qr(t) = r(t+1)$.
We refer to (\ref{eq:S_internal}) as the network equation and to  (\ref{eq:S_external}) as the measurement equation.
The matrices $B \in \mathbb{Z}_2^{n \times m}$ and $C \in \mathbb{Z}_2^{p \times n}$ are binary selection matrices with full column rank and full row rank, respectively. 
They select which external signals are applied to the nodes and which measurements are taken from the network.
Let $\mathcal{V}$ be the set of all nodes, and $\mathcal{B}$ and $\mathcal{C}$ the sets of excited and measured nodes, respectively.
These two sets define the EMP, formally introduced in \cite{MapurungaBazanellaOptimalAllocationExcitation2021}.
\begin{definition}
    A pair of selection matrices $B$ and $C$, with its corresponding pair of node sets $\mathcal{B}$ and $\mathcal{C}$,  defines  an \textbf{Excitation and Measurement Pattern (EMP)}.
		An EMP is called \textbf{valid} for the network (\ref{eq:S_internal})-(\ref{eq:S_external}) if this network is generically identifiable 
		\footnote{See \cite{hendrickx-gevers-bazanella-identifiability-2019} for a formal definition.}
		with this EMP.
  Let $\nu = |\mathcal{B}| + |\mathcal{C}|$ \footnote{$|\cdot|$ - Denotes the cardinality of a set.} be the cardinality of an EMP.
  A given EMP is called  \textbf{minimal} for this network if it is valid and there is no other valid EMP with smaller cardinality.
  \label{def:EMP}
\end{definition}
\vspace{-2mm}
%
Associated with the dynamic network there is a digraph $\mathcal{G}(\mathcal{V}, \mathcal{E})$ where $\mathcal{V}$ and $\mathcal{E}$ stand for the set of nodes and edges, respectively. 
Each node $i \in \mathcal{V}$ is associated with an internal signal $w_i(t)$, and each edge $(i, j)$  is associated with a module: $a^0_{ji} q^{-1}$, where $(i, j) \in \mathcal{E}$, if and only if $a^0_{ji} \neq 0$.
For an edge $(i, j)$, the node $i$ is an in-neighbor of node $j$ ($j$ is an out-neighbor of $i$), and we say $(i, j)$ is an outgoing edge of $i$ (ingoing edge of $j$).   
A source (sink) is a node with no in-neighbors (out-neighbors), while nodes that are neither are called \textit{internal}. 
Let $\mathcal{F}$, $\mathcal{S}$, and $\mathcal{I}$ denote the sets of sources, sinks, and internal nodes, respectively. 
A path $\mathcal{P}_{ji}$ from node $i$ to node $j$ is a sequence of edges $((i, k_1), (k_1, k_2), \dots, (k_n, j))$. 
A digraph is weakly connected if a path exists between any pair of nodes. 

In this paper, we are interested in a specific class of networks, namely trees, formally defined in the following.

\begin{definition}
				A tree is any digraph that has no cycles even if one changes the direction of the edges. 
				\label{def:tree}
\end{definition}
\vspace{-2mm}
Trees are a natural extension of branches; 
a branch is just a line (or a path) from a tree. 
A consequence of Definition~\ref{def:tree} is that any tree with $n$ nodes has exactly $n-1$ edges. 
Another consequence is that there is at most one path between any two nodes in a tree. 

In this work, we are interested in obtaining the most accurate minimal EMP for a tree.
The smallest cardinality any EMP can achieve is $n$, see \cite{BazanellaEtAlNetworkIdentificationPartial2019a}. 
For trees, necessary and sufficient conditions for valid (minimal) EMPs are given in the following Proposition.
\begin{proposition}\cite{BazanellaEtAlNetworkIdentificationPartial2019a}
				All transfer functions in a tree network are generically identifiable if only if $\mathcal{F} \subseteq \mathcal{B}, \mathcal{S} \subseteq \mathcal{C}, {\mathcal B} \cup {\mathcal C} = {\mathcal V}$.\label{teo:ExciteOrMeasure}
\end{proposition}
\vspace{-2mm}
This result means that all sources from a tree must be excited, all sinks must be measured, and every internal node must be either excited or measured. 
Proposition \ref{teo:ExciteOrMeasure} defines all minimal EMPs for any tree network. 

\section{Network Identification}
\label{sec:Problem}



Designing a minimal EMP is an experiment design problem. 
One objective that should guide this decision is to choose the minimal EMP that yields most accurate estimates.
In this work, we adopt a maximum likelihood approach through the application of the Prediction Error Method (PEM) to the network model (\ref{eq:S_internal})-(\ref{eq:S_external}).
The PEM consists of minimizing the Prediction Error  $\varepsilon (t) := y(t) - \hat{y}(t|t-1, \theta)$, where $\theta$ is vector of parameters of the network, i.e. those from the network matrix $A(\theta)$.
The optimal one-step predictor for it is obtained as follows
\begin{equation}
  \hat{y}(t|t-1, \theta) = C (I - A(\theta))^{-1} B r(t) := C T(\theta) B r(t).\label{eq:optpred}
\end{equation}
Under mild hypothesis, the accuracy of PEM achieves asymptotically the Cramer-Rao lower bound for the covariance matrix. 
Accordingly, the accuracy of the minimal EMPs will be assessed through the Cramer-Rao lower bound.
This matrix can be calculated as
$P = 1 / N[\mathbb{E} \psi(t)\Lambda^{-1}\psi^T(t)]^{-1}$, where $\psi(t)$ is the gradient of the optimal one-step ahead predictor, $\Lambda$ is the noise's covariance matrix, and $N$ is the number of data samples
\citep{ljung1998system}. 
The covariance matrix can also be obtained as the inverse of the information matrix $M := N \mathbb{E} \psi(t) \Lambda^{-1} \psi^T(t)$.
This matrix is essential for assessing the accuracy of the parameter estimates. 
The overall accuracy is assessed by the trace of the matrix $P$, which is known as the A-optimality criterion  \citep{pukelsheimOptimalDesignExperiments2006}.
We consider that the following assumptions hold within the prediction error framework:
\begin{enumerate}[label=(\alph*)]
	\item the dynamic network is stable; 
	\item the signals $\lbrace r_j(t)\rbrace$ are mutually independent stationary white noise processes with zero mean and variance $\sigma_j$; they are uncorrelated with all processes $\left\lbrace e_j\right\rbrace$; \label{item:input-white-noise}
	\item the corrupting noise sequences $\left\lbrace e_j\right\rbrace $ are independent stationary Gaussian white noise processes with zero mean and variance $\lambda_j$.  \label{item:noise-lambda}
\end{enumerate}
Before presenting the analysis for trees, 
let us define some additional nomenclature.
\begin{definition}
	A dynamic network is said to be \emph{uniformly excited} if $\lambda_j = \lambda$ and  $\sigma_i = \sigma$, for all $ j \in \mathcal{C}$ and $i \in \mathcal{B}$.
	It is  a \emph{fully symmetric} network when it is uniformly excited and 
	all modules of the network have the same magnitude.
	\label{def:UE} 
\end{definition}
\vspace{-2mm}
The fully symmetric scenario might be unrealistic at first for a meaningful analysis. 
However, it has been shown in \cite{MapurungaBazanellaOptimalAllocationExcitation2021, mapurunga-bazanella-optimal-excitation-2023} that the insights obtained for this case remain valid also for the more general case. 
Furthermore, this scenario sheds light on the role of the network topology plays into the accuracy of the parameter estimates. 

In the next sections we are going to show how to compute the information matrix for general networks and further we specialize into state space tree networks.

\section{Information Matrix}
\label{sec:infmat}

The information matrix is fundamental for accuracy analysis and experiment design.
For the optimal one-step predictor defined in (\ref{eq:optpred}), the information matrix depends on the gradient of the optimal predictor. 
In order to compute the gradient, let us decompose it as follows:
\begin{align}
\psi(t, \theta) =  \Big{[}\sum_{i \in \mathcal{B}} \psi_{c_1 i}(t, \theta) ~ \quad \cdots \quad
&  \sum_{i \in \mathcal{B}} \psi_{c_p i}(t, \theta)\Big{]} 
\label{eq:psidec}
.\end{align}
%
for $\{c_1~ c_2~ \dots~ c_p\} = \mathcal{C}$, with
$\psi_{ji}(t, \theta) \triangleq \frac{\partial T_{ji}(q, \theta)}{\partial \theta} r_i(t)$  a column-vector, and $T_{ji}(q, \theta)$ the $(j, i)$ element from $T(q, \theta)$. 
In order to isolate the information content from a particular node to another, we are going to introduce the concept of \emph{partial} information matrix corresponding to the input from node $i$ to the output of node $j$. 
%
\begin{definition}
				The partial information matrix associated with the input applied in node $i$ and the measurement at node $j$ is:
				\begin{align}
							M_{ji} = {N}/{\lambda_j} \mathbb{E} \psi_{ji}(t, \theta) \psi_{ji}^T(t, \theta)
							\label{eq:defMji}
				.\end{align}
				\label{def:partialMji}
\end{definition}
\vspace{-4mm}
%
With this definition at hand and from the technical assumptions (b) on the input and  (c) on the external noise, the information matrix can be computed as the sum of the partial information matrices defined by the EMP.
This is formally stated in the following Lemma.
\begin{lemma}
  The information matrix of a dynamic network (\ref{eq:S_internal})-(\ref{eq:S_external}) under assumptions (a)-(c) can be computed as the sum of partial information matrices defined in (\ref{eq:defMji}): 
				\begin{align}
							M = \sum_{j \in \mathcal{C}, i \in \mathcal{B}} M_{ji}
							\label{eq:infmat}
				.\end{align}
				\label{lem:infmat}
\end{lemma}
\vspace{-7mm}
\begin{pf}
				Under the stated assumption the information matrix can be computed as $M:= N \mathbb{E} \psi(t) \Lambda^{-1} \psi^T(t)$. 
				Using (\ref{eq:psidec}) and the fact that $\Lambda = diag(\lambda_{c_1}, \lambda_{c_2}, \dots \lambda_{c_p})$, we can decompose $M$ as:
				\begin{align*}
								M =  \frac{N}{\lambda_{c_1}} \mathbb{E} \sum_{i \in \mathcal{B}} \psi_{c_1, i}(t, \theta) \sum_{i \in \mathcal{B} } \psi_{c_1, i}^T(t, \theta) +
								\cdots  + \\
								\frac{N}{\lambda_{c_p}}\mathbb{E} \sum_{i \in \mathcal{B}} \psi_{c_p, i}(t, \theta) \sum_{i \in \mathcal{B} } \psi_{c_p, i}^T(t, \theta) 
				.\end{align*}
				Since all inputs are mutually independent, it holds that the terms $\mathbb{E} \psi_{k, i}(t, \theta) \psi_{l, j}^T (t, \theta) = 0$, 
				for $k, l \in \mathcal{C}$ and $i \neq j \in \mathcal{B}$. 
				The remaining nonzero terms are the ones in (\ref{eq:defMji}).
\end{pf}
\vspace{-3mm}
This Lemma shows that we can compute the information matrix through the sum of the partial information matrices.
				This result is valid for any dynamic network subject to the standing assumptions (a)-(c).
We are now going to focus in obtaining the expressions for each partial information matrix of tree networks. 


\section{Computing the Information Matrix for tree networks}
\label{sec:compinfmat}

For trees, the transfer function $T_{ji}(q)$ is nonzero
if and only if there is a path from node $i$ to node $j$. 
Let $\mathcal{M}_{ji} := (a_{i, k_1}, a_{k_1, k_2}, \dots, a_{k_p, j})$ represent the modules in the path from node $i$ to $j$ and let $n_{p_{ji}}$ be the length of path $\mathcal{P}_{ji}$, i.e. $n_{p_{ji}} := |\mathcal{P}_{ji}|$, then the transfer function $T_{ji}$ is the product of the modules in the path
\begin{align}
				T_{ji} = q^{ -n_{p_{ji}}}
				\prod_{a \in \mathcal{M}_{ji}} a .
				\label{eq:Tji}
\end{align}
In order to obtain the gradient of $T_{ji}$ with respect to $a_{kl}$ we have two scenarios as follows
\begin{align}
				\frac{\partial T_{ji}}{\partial a_{kl}} =
				\begin{cases}
								r_i(t) q^{-n_{p_{ji}}}\frac{1}{a_{kl}}  \prod_{a \in \mathcal{M}_{ji}} a, &\text{if}\, a_{kl} \in \mathcal{M}_{ji}. \\
								 0, &\text{if}\, a_{kl} \not\in \mathcal{M}_{ji}.
				\end{cases}
        \label{eq:dTjidakl}
\end{align}
Some elements of the gradient are zeros associated with the parameters $a_{kl}$ that are not part of the path from node $i$ to node $j$. 
Let $\mathcal{L}_{ji}$ be the tuple of indexes associated with the parameters $a_{kl}$ in $\mathcal{M}_{ji}$ from path $\mathcal{P}_{ji}$.
The elements of the gradient associated with those parameters are:
%
\begin{align}
				\psi_{ji}[l] = r_i(t-n_{p_{ji}}) \frac{1}{a_{l}} \prod_{k \in \mathcal{L}_{ji}} a_k 
				\label{eq:psiji}
,\end{align}
where $[l]$ for $l \in \mathcal{L}_{ji}$ stands for those elements associated with parameters that are in the path $\mathcal{P}_{ji}$, 
 and $a_k$ stands for the $k$-th element of $\mathcal{M}_{ji}$. 
The partial information matrix can be obtained from the gradient.
Thus, the elements of the partial information matrix can be obtained as
\begin{align}
				M_{ji}[k, l] &= \frac{N}{\lambda_j} \mathbb{E} \psi_{ji}[k]\psi_{ji}[l]^T  
										 = \frac{N \sigma_i}{\lambda_j} \frac{1}{a_k a_l} \prod_{m \in \mathcal{L}_{ji}} a_m^2 
				\label{eq:Mji}
,\end{align}
where $M_{ji}[k, l]$ stands for the ($k, l)$ element of the partial information matrix $M_{ji}$ for $k, l \in \mathcal{L}_{ji}$. 
All other elements will be zero. 
One can easily obtain the information matrix from the partial information matrices associated with an EMP, as will be illustrated in the next section.

\section{Exciting or measuring?}
\label{sec:exciteormeasure}

Any valid EMP must have all sources excited and all sinks measured. 
For a tree, any choice of excitation or measurement of the internal nodes  
will result in a valid and possibly minimal EMP \citep{BazanellaEtAlNetworkIdentificationPartial2019a}. 
So, in a tree the choice of the best (i.e. most accurate) EMP boils down to 
choosing, for each internal node, whether to excite or to measure it. 
In this section we examine this question: whether we should excite or measure a given internal node. 
In order to gain insights into this problem, we will consider two simple examples of specific tree topologies, which we call a cross and inverted cross. 
Two five nodes cross and inverted cross are depicted in Figure \ref{fig:Icross}. 
\begin{figure}[h!]
    \centering

    \begin{minipage}{0.25\columnwidth}
        \centering
        \includegraphics[width=\linewidth]{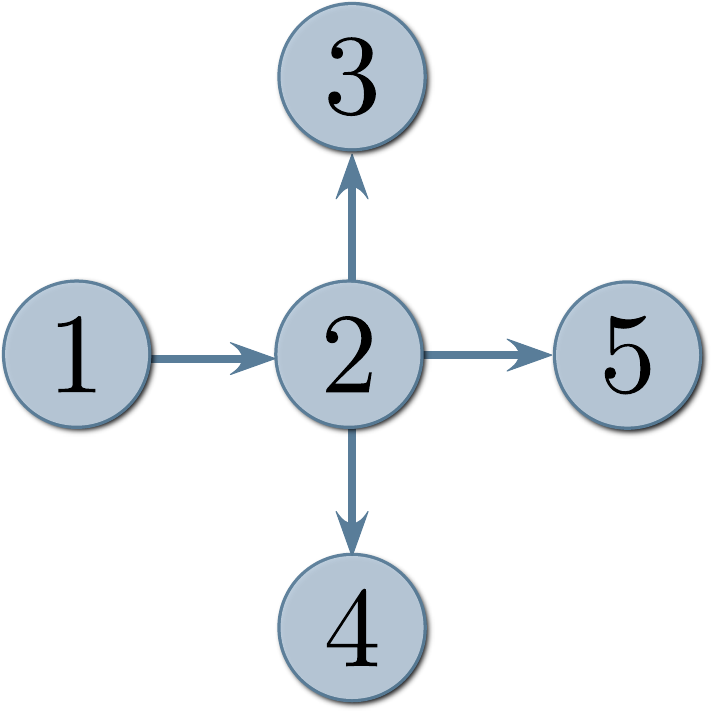}

        \small (a) A 5-node cross
    \end{minipage}
    \hspace{10pt}
    \begin{minipage}{0.25\columnwidth}
        \centering
        \includegraphics[width=\linewidth]{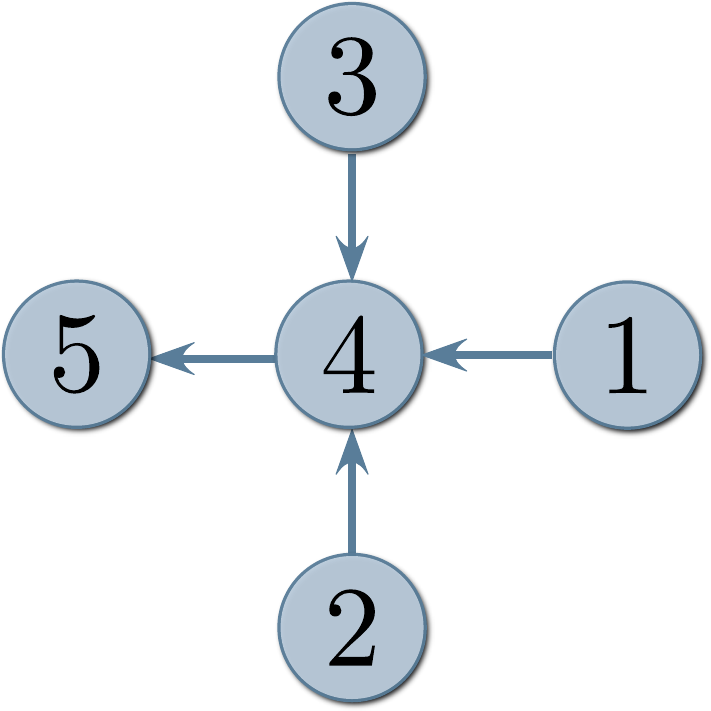}

        \small (b) A 5-node inverted cross
    \end{minipage}

    \caption{Two classes of tree networks.}
    \label{fig:Icross}
\end{figure}

%
%
The network matrix for the cross and the inverted cross are given by
\begin{align}
				A_c = 
				\begin{bmatrix}
            0 & 0 & 0 & 0 & 0 \\
            a_{21}^0 & 0 & 0 & 0 & 0 \\
            0 & a_{32}^0 & 0 & 0 & 0 \\
            0 & a_{42}^0 & 0 & 0 & 0 \\
            0 & a_{52}^0 & 0 & 0 & 0 \\
				\end{bmatrix},
				A_i^T =
				\begin{bmatrix} 
            0 & 0 & 0 & a_{41}^0 & 0 \\
            0 & 0 & 0 & a_{42}^0 & 0 \\
            0 & 0 & 0 & a_{43}^0 & 0 \\
            0 & 0 & 0 & 0 & a_{54}^0 \\
            0 & 0 & 0 & 0 & 0 \\
				\end{bmatrix}
				\label{eq:netcross}
.\end{align}

Let us start our analysis with the cross network. 
Notice that node $1$ is a source, while nodes $3, 4, 5$ are sinks.
From the necessary and sufficient conditions for trees given in Proposition~\ref{teo:ExciteOrMeasure}, there are only two minimal EMPs:
$\emp{I} : \left(\mathcal{B} = \{1, 2\}, \mathcal{C} = \{3, 4, 5\}\right)$ and
$\emp{II} : \left(\mathcal{B} = \{1\}, \mathcal{C} = \{2, 3, 4, 5\}\right)$.
The difference between these two minimal EMPs is whether we decide to excite or measure node 2. 
There are four parameters to be estimated in this network, namely, $a_{21}, a_{32}, a_{43}, a_{54}$. 
The information matrix associated with $\emp{I}$ is easily obtained by exploring the concept of partial information matrix:
%
\begin{align*}
				&M_I = M_{31} + M_{32} + M_{41} + M_{42} + M_{51} + M_{52} = \\ 
				&\resizebox{.9\hsize}{!}{$\left[\begin{matrix}
												\sigma_1\left[\frac{a_{52}^{2}}{\lambda_{5}} + \frac{a_{42}^{2}}{\lambda_{4}} + \frac{ a_{32}^{2}}{\lambda_{3}}\right] & 
                        {\sigma_{1} a_{21} a_{32}}/{\lambda_{3}} &
                        {\sigma_{1} a_{21} a_{42}}/{\lambda_{4}} & 
                        {\sigma_{1} a_{21} a_{52}}/{\lambda_{5}}\\
                        {\sigma_{1} a_{21} a_{32}}/{\lambda_{3}} & {[\sigma_{1} a_{21}^{2} + \sigma_2]}/{\lambda_{3}} & 0 & 0\\
                {\sigma_{1} a_{21} a_{42}}/{\lambda_{4}} & 0 & {[\sigma_{1} a_{21}^{2} + \sigma_2]}/{\lambda_{4}}  & 0\\
                {\sigma_{1} a_{21} a_{52}}/{\lambda_{5}} & 0 & 0 & {[\sigma_{1} a_{21}^{2} + \sigma_2]}/{\lambda_{5}} 
\end{matrix}\right]$}
.\end{align*}
whereas for \emp{2} we have:
\begin{align*}
				M_{II} = M_{21} + M_{31} + M_{41} + M_{51}
.\end{align*}
The inverse of the information matrix will provide the asymptotic accuracy of the parameter estimates. 
In order to compare these minimal EMPs, we will consider the effect on the accuracy of these EMPs on the estimates of $a_{21}$.
The asymptotic accuracy of $\hat{a}_{21}$ for \emp{I} and \emp{II} are given as follows
\begin{align}
				{\rm{cov}}(\hat{a}_{21}^{I}) &=
\frac{\lambda_{3} \lambda_{4} \lambda_{5} \left(\sigma_{1} a_{21}^{2} + \sigma_{2}\right)}{\sigma_{1} \sigma_{2} \left(\lambda_{3} \lambda_{4} a_{52}^{2} + \lambda_{3} \lambda_{5} a_{42}^{2} + \lambda_{4} \lambda_{5} a_{32}^{2}\right)}
\label{eq:cova21I}
\\
				{\rm{cov}}(\hat{a}_{21}^{II}) &= {\lambda_2}/{\sigma_1} \label{eq:cova21II}
.\end{align}
where the superscript $I,II$ stands for the respective EMP. 
Notice that for \emp{II}, the covariance of $\hat{a}_{21}$ only depends on the SNR applied from excitation of node $1$ and measurement of node $2$.
This phenomenon has already been observed for branches in 
\cite{mapurunga-bazanella-optimal-excitation-2023}. 
When node $2$ is measured, the covariance of $\hat{a}_{21}$ does not depend on the other measurements from the network, 
				as expected.
This means that only the adjacent nodes have an influence on the accuracy of $\hat{a}_{21}$ provided node $2$ is measured. 
With respect to the accuracy, whether to excite or measure node 2, we have the following result.

\begin{proposition}
  Consider a cross network defined by the network matrix $A_c$ in (\ref{eq:netcross}). For a fully symmetric network (Definition \ref{def:UE}), the asymptotic estimate of $a_{21}$ given by \emp{I} (exciting node 2) will be more accurate than \emp{II} counterpart if and only if
				$a_{21}^2 > 0.5$. 
\label{prop:cross}
\end{proposition}
\vspace{-4mm}
\begin{pf}
				Under the stated assumption, we have the following expression from (\ref{eq:cova21I})-(\ref{eq:cova21II})
				\begin{align}
								{\rm{cov}}(\hat{a}_{21}^I) &= \left(  {\lambda_2} {(a_{21}^2 + 1)}\right) /\left(  {\left(a_{32}^2 + a_{42}^2 + a_{52}^2  \right) {\sigma_1} }\right) \\
								{\rm{cov}}(\hat{a}_{21}^{II}) &= {\lambda_2}/{\sigma_1} 
				\end{align}
				The covariance of $\hat{a}_{21}$ will be more accurate in \emp{I} if and only if
        				$a_{21}^2 + 1 < a_{32}^2 + a_{42}^2 + a_{52}^2$. 
				For a fully symmetric network, we have that 
								$a_{21}^2 + 1  < 3 a_{21}^2 \iff a_{21}^2 > 0.5.$  
\end{pf}
\vspace{-4mm}
This result shows that the choice between exciting or measuring a node with respect to the accuracy of a module depends,
among other things, on the module magnitudes. 

A similar analysis can be made for the accuracy of $a_{32}, a_{42}$ and $a_{52}$. 
Due to the symmetry, the analysis for any of these modules will be virtually the same. 
For this reason, we are going to focus on the accuracy of $a_{52}$. 
One would expect that by choosing \emp{I}, this would imply that the covariance of $\hat{a}_{52}$ would be just the inverse of the SNR applied at the corresponding nodes.
However, differently from the branch case, this does not happen here. 
In fact, we have that
\begin{align}
				{\rm{cov}}({\hat{a}^I_{52}}) &= \frac{\lambda_{5} 
				\left[\lambda_{3} \lambda_{4} a_{52}^2 [\sigma_{1} a_{21}^{2} + \sigma_{2}] +  \sigma_{2} [\lambda_3 a_{42}^{2} + \lambda_{4} a_{32}^{2}]\right]}
				{\sigma_{2} \left[\sigma_{1} a_{21}^{2} + \sigma_{2}\right] \left[\lambda_{3} \lambda_{4} a_{52}^{2} + \lambda_{3} \lambda_{5} a_{42}^{2} + \lambda_{4} \lambda_{5} a_{32}^{2}\right]} \nonumber\\
								{\rm{cov}}\left(\hat{a}^{II}_{52} \right) &= {\left(  \lambda_{2} a_{52}^{2} + \lambda_{5}\right) }/{\left(  \sigma_{1} a_{21}^{2}\right) } \label{eq:a52sur}
.\end{align}
This differs from the branch case, where one would expect that the covariance of $\hat{a}^{I}_{52}$ does depend on all the parameters from the network. 
As for the selection of exciting or measuring node 2 with respect to the accuracy of all parameters in the network, we have the following result.
\begin{theorem}
  Consider the cross defined by the network matrix $A_c$ in (\ref{eq:netcross}). For a fully symmetric network (Definition~\ref{def:UE}), \emp{I} yields the most accurate estimates with respect to the trace of the asymptotic covariance matrix.
				\label{theo:cross}
\end{theorem}%
\vspace{-5mm}
\begin{pf}
				Under the stated assumptions, we have that $a_{21} = a_{32} = a_{42} = a_{52}$.
				We then have that \emp{I} is most accurate than \emp{II} if and only if
								${\rm{tr}(P_{I})}/{\rm{tr}(P_{II})} < 1$. 
				This is equivalent to
				\begin{equation*}
								\frac{4 a_{21}^{4} + 11 a_{21}^{2} + 1}{3 \left(a_{21}^{2} + 1\right) \left(4 a_{21}^{2} + 3\right)} < 1 
								\iff
								8 a_{21}^4 +17 a_{21}^2 + 8 > 0. 
                \tag*{$\blacksquare$}
				\end{equation*}
        \renewcommand{\qedsymbol}{}
\end{pf}
\vspace{-5mm}
This result indicates that, when all factors are equal, exciting node $2$ is preferable to measuring it.
We should expect, as in the branch case, that excitation of node $2$ will result in most accurate estimates most of the times under different scenarios, as evidence suggested by the results for branches 
\citep{MapurungaBazanellaOptimalAllocationExcitation2021, mapurunga-bazanella-optimal-excitation-2023}. 

Let us analyze now the inverted cross, 
defined by $A_i$ in (\ref{eq:netcross}).
The parameters to be identified are $a_{41}, a_{42}, a_{43}$ and $a_{54}$. 
As in the cross network, there are two minimal EMPs for this network, namely, 
\emp{I}: $\left(\mathcal{B} = \{1, 2, 3, 4\} , \mathcal{C} = \{5\} \right) $ and
\emp{II}: $\left(\mathcal{B} = \{1, 2, 3\} , \mathcal{C} = \{4, 5\} \right) $.
The difference between these two minimal EMPs relies on whether we decide to excite or measure the node $4$.

First, let us consider the accuracy of $\hat{a}_{54}$, this parameter is the dual of the cross network parameter $a_{21}$.
The covariances for both minimal EMPs are given as: 
\begin{align}
				{\rm{cov}}\left( \hat{a}^{I}_{54} \right) &= {\lambda_5}/{\sigma_4} \label{eq:cova54I}	 \\ 
				{\rm{cov}}\left( \hat{a}^{II}_{54} \right) &={\left(\lambda_{4} a_{54}^{2} + \lambda_{5}\right)}/{\left(\sigma_{1} a_{41}^{2} + \sigma_{2} a_{42}^{2} + \sigma_{3} a_{43}^{2}\right)} \label{eq:cova54II}	
.\end{align}
As expected, for \emp{I}, the covariance of $\hat{a}_{54}$ is just the inverse of SNR, whereas for \emp{II} is a function of all network parameters. 
Let us now consider the covariance of $\hat{a}_{41}$, which will be similar to those of $\hat{a}_{42}$ and $\hat{a}_{43}$. 
The covariances are given in what follows.
\begin{align*}
				{\rm{cov}}(\hat{a}^I_{41}) &= {\left(\lambda_{5} \sigma_{1} a_{41}^{2} + \lambda_{5} \sigma_{4}\right)}/{\left(\sigma_{1} \sigma_{4} a_{54}^{2}\right)}\\
				{\rm{cov}}(\hat{a}^{II}_{41}) &= 
				\frac{\lambda_{4} \left(\lambda_{4} \sigma_{1} a_{41}^{2} a_{54}^{2} + \lambda_{5} (\sigma_{1} a_{41}^{2} +  \sigma_{2} a_{42}^{2} + \sigma_{3} a_{43}^{2})\right)}
				{\sigma_{1} \left(\lambda_{4} a_{54}^{2} + \lambda_{5}\right) \left(\sigma_{1} a_{41}^{2} + \sigma_{2} a_{42}^{2} + \sigma_{3} a_{43}^{2}\right)}
.\end{align*}
Similar to the cross network case, the covariance of $\hat{a}_{41}$ depends on all parameters of the network. 
This phenomenon is somehow remarkable since one would expect that $a_{42}$ and $a_{43}$ should not affect the accuracy of $\hat{a}_{41}$ from the network topology.  
As in the cross network case, we have the following results regarding the inverted cross. 

\begin{proposition}
				Consider the inverted cross defined by equations. 
				For a fully symmetric network, see Definition~\ref{def:UE}, the asymptotic estimates of ${a}_{54}$ will be more accurate for \emp{II} with respect to \emp{I} if and only if
					$a_{54}^2 > 0.5$			
				\label{prop:invcrossa54}
\end{proposition}
\vspace{-4mm}
\begin{pf}
				For a fully symmetric network we have that $a_{41} = a_{42} = a_{43} = a_{54}$. 
				\emp{II} yields a more accurate  covariance of $\hat{a}_{54}$ if and only if
$3 a_{54}^2 > a_{54}^2 + 1 \iff a_{54}^2 > 0.5.$ 
\end{pf}
\vspace{-3mm}
This result is the dual of Proposition \ref{prop:cross}. 
For a fully symmetric inverted cross, it is better to measure node $4$ if the magnitude of the parameter is big. 
As for the accuracy of all parameters, the dual result is the following. 

\begin{theorem}
				Consider the inverted cross defined by $A_i$ in (\ref{eq:netcross}). 
				For a fully symmetric network (see Definition~\ref{def:UE}), \emp{II} will be more accurate than \emp{I} with respect to the trace of the covariance matrix.
				\label{theo:icross}
\end{theorem}
\vspace{-5mm}
\begin{pf}
			The proof is dual to Theorem \ref{theo:cross}. 
\end{pf}
\vspace{-4mm}
This section suggests that it is better to measure internal nodes for which there are many excited in-neighbors, 
whereas it is better to excite internal nodes in which there are many out-neighbors. 
Furthermore, for network parameters with large magnitude, it is better to measure the in-neighbors of sinks and excite the out-neighbors of sources.  
In addition, from the duality observed by the cross and inverted cross, one should excite nodes with many out-neighbors and measure nodes with many in-neighbors. 
What the study of the crosses reveals is the following.
\begin{enumerate}
\item To optimize the precision of the overall identification:
	\begin{enumerate}
	\item nodes with large in-degrees should be measured;
	\item nodes with large out-degrees should be excited;
	\item EMPs that results in a larger number of identified $T$'s should be preferred. 
	\end{enumerate}
\item To otimize the precision of a particular edge:
	\begin{enumerate}
	\item the best choice depends on the magnitude of the particular edge and sometimes on the magnitudes of the other edges as well;
	\item direct identification (that is, excite the start and measure the end of the edge) is best if the edge's magnitude is small. 
	\end{enumerate}	
\end{enumerate}

      This guideline should be followed on top of the one for branches: excite nodes near the sources and measure nodes near the sinks. 

\section{When does only the SNR matter?}
\label{sec:justSNR}

In \cite{mapurunga-bazanella-optimal-excitation-2023}, it was shown that, for branch networks, measuring node $2$ renders the accuracy of the first module independent of the remaining network; the dual result was also established.
As shown in the previous section, this property only partially extends to tree networks.
Now, we provide a sufficient condition characterizing this phenomenon for tree networks. 

\begin{theorem}
				Consider any directed tree. 
				Let $a_{ji}q^{-1}$ be a module of interest. 
				We have the following.
				\begin{enumerate}
								\item Let $i$ be a source and $j$ its unique out-neighbor. 
								If $i$ is the only in-neighbor of $j$, then measuring $j$ renders the accuracy of $\hat{a}_{ji}$ independent of all other modules; 
								\item Let $j$ be a sink and $i$ its unique in-neighbor.
								If $j$  is the only out-neighbor of $i$, then exciting $i$ renders the accuracy of $\hat{a}_{ji}$ independent of all other modules;
				\end{enumerate}
				In both cases, the accuracy of $\hat{a}_{ji}$ will be the inverse of the SNR: ${\rm{cov}}(\hat{a}_{ji}) \sim {\lambda_j}/{\sigma_i}$.
\end{theorem}
\vspace{-4mm}
\begin{pf}
				We are going to focus on the first case, the second case follows from duality. 
				Assume, without loss of generality, that $a_{21}$ is the parameter of interest and let $3, 4, \dots, k$ be 
				the out-neighbors of node $2$ with node $1$ being its only in-neighbor. 
				Assume that all sources are excited, and all other nodes are measured.  
				Now, one can build the information matrix from the partial information matrices $M_{ji}$ as in (\ref{eq:Mji}) for $i \in \mathcal{F}$ and $j \in \mathcal{V}\setminus\mathcal{F}$. 
				We are interested in determining the first $k$ rows of the information matrix, which can be obtained as
				\begin{align*}
								M = &M_{21} + M_{31} + \cdots + M_{k1} + M_{k+1, 1} + \cdots + M_{n 1}  + \\&\sum_{j \in \mathcal{V}\setminus \mathcal{F}, i \in \mathcal{F}\setminus \{1\}} M_{ji}
				.\end{align*}
        The other sources will not make an effect on the first $k$ rows. This implies 
								$M[1, 1] = {\sigma_1}/{\lambda_2} + X$, where   
				$X$ is the sum of the corresponding elements from $M_{21}$, $M_{32}, \dots$. 
				Now, consider the permutation matrix
				\begin{align*}
								\rho := 
                \left[
								\begin{smallmatrix}
												1 & -{a_{32}}/{a_{21}} & -{a_{42}}/{a_{21}} & \cdots & -{a_{k,2}}/{a_{21}} & 0 & \cdots & 0 \\
												0 &  &  & & I \\
                    \end{smallmatrix} \right]
				.\end{align*}
        The covariance matrix can be obtained as
\[
P=\rho^\top\bar M^{-1}\rho\!=\!
\left[\begin{smallmatrix}
\lambda_2/\sigma_1&\\&\tilde S
\end{smallmatrix}\right],\
\bar M:=\rho M\rho^\top\!=\!
\left[\begin{smallmatrix}
\sigma_1/\lambda_2&0\\0&S
\end{smallmatrix}\right]. 
\tag*{$\blacksquare$}
\]
       \renewcommand{\qedsymbol}{}
\end{pf}
\vspace{-6mm}
This result shows that sources with a single outgoing edge, 
if its out-neighbor is measured, then the asymptotic covariance of the corresponding module will not depend on any parameter of the network. 
The same holds for sinks with a single incoming edge. 
As we have seen from the previous examples, for large magnitude of the network parameters, a most accurate estimate can be obtained by doing just the opposite, that is, by exciting the out-neighbor of a source and by measuring the in-neighbor of a sink. 
This is in line with the results for branches derived in 
\cite{MapurungaBazanellaOptimalAllocationExcitation2021, mapurunga-bazanella-optimal-excitation-2023}. 
Furthermore, this result explains the relationship obtained in equations (\ref{eq:a52sur}) and (\ref{eq:cova54I}) from the cross and inverted cross.

\section{All five-nodes trees}
\label{sec:21trees}

In the previous Section, we have provided Theorems that are valid for a specific class of trees, that we named crosses.
We have gained insight into the problem on whether we should excite or measure the internal nodes.  
In this Section we examine by means of various case studies to which extent these principles extend to other tree topologies.
Specifically, we test the insights for all possible tree networks with five nodes. 
The crosses in the previous sections are one particular 5-node tree topology, among 21 different tree topologies with 5 nodes that have at least one internal node.
Trees without internal nodes have only one minimal EMP (exciting the sources and measuring the sinks), therefore they are of no interest here. 
The edges from the 21 trees are depicted in Table \ref{tab:21trees}.
The theoretical information matrices will be calculated for all the resulting 21 trees, this allows to check the principles exhaustively for all 5-node tree topologies, and it is
what we'll present next.

\begin{table}[htpb]
        \centering
        \caption{All 21 five node trees with at least one internal node.}

				\label{tab:21trees}

{\tiny
\setlength{\tabcolsep}{1pt}
\renewcommand{\arraystretch}{1}

\resizebox{\columnwidth}{!}{%
\begin{tabular}{c|c|c|c|c|c}
\hline
Tree & edges & Tree & edges & Tree & edges\\ \hline

1 & $(1,2), (1,3), (1,4), (2,5)$ &
2 & $(1,2), (1,3), (2,4), (2,5)$ &
3 & $(1,2), (1,3), (2,4), (3,5)$ \\

4 & $(1,3), (1,4), (3,2), (2,5)$ &
5 & $(1,3), (1,2), (2,5), (4,5)$ &
6 & $(1,2), (2,3), (2,4), (2,5)$ \\

7 & $(1,2), (2,3), (2,4), (3,5)$ &
8 & $(1,2), (2,5), (5,3), (5,4)$ &
9 & $(1,3), (3,5), (5,4), (4,2)$ \\

10 & $(1,3), (1,4), (3,2), (5,4)$ &
11 & $(1,2), (2,3), (2,4), (5,4)$ &
12 & $(1,3), (3,5), (2,4), (5,4)$ \\

13 & $(1,2), (1,3), (2,5), (4,2)$ &
14 & $(1,2), (2,3), (5,3), (5,4)$ &
15 & $(1,3), (3,2), (2,5), (4,2)$ \\

16 & $(1,3), (3,5), (2,5), (4,2)$ &
17 & $(1,3), (3,2), (4,2), (5,2)$ &
18 & $(1,2), (2,3), (4,2), (2,5)$ \\

19 & $(1,2), (4,2), (2,3), (3,5)$ &
20 & $(1,2), (4,2), (2,5), (3,5)$ &
21 & $(2,1), (5,2), (3,2), (4,2)$ \\

\hline
\end{tabular}%
}
}
\end{table}

For each topology we analyze three scenarios: 
(1) uniform excitation with small random parameter values; 
(2) uniform excitation with large random parameter values;
(3) fully random - that is, random $\lambda_i$, $\sigma_j$ and $a_{ji}$.
%
In each scenario, 10,000 cases are randomly generated. 
For each case, we assess the precision of all minimal EMPs, determining what is the best EMP (i.e. the one resulting in the smallest trace).
Our interest is to determine how often the EMP generated by our principles - ``our EMP'' - is the best. 

To clarify, consider Tree 7, for which there are 4 minimal EMPs. 
We generated 10,000 fully random cases (scenario 3) and evaluate what is the best EMP in each case. 
Our EMP was the best in 7,333 cases ($73.3 \%$), which also makes it the champion - that is, there is no other EMP that is best in a larger number of cases. 
For the same topology, in scenario 1 our EMP is the best in ``only'' $49.9\%$ of cases,
but it is still the champion, as there is no other EMP that is best in a larger number of cases.

    Figure \ref{fig:radar} depicts the percentage of the best EMP, our EMP and the runner up EMP for each scenario and for the trees that we have specific suggestions. 
		In all cases, for all trees,``our EMP" is the champion with a large difference for the runner up EMP in most cases, except for trees 9 and 19.  
		These trees are the inverse of one another (that is, if we invert the edges of one we obtain the other).
		For these trees we have a conflicting suggestion. 
		By one hand we have a branch and we must excite nodes near the sources and measure nodes near the sinks,
		by another hand, we have to measure nodes with multiple incoming edges, and excite nodes with multiple outgoing edges. 
		The same holds for Tree 9, which is just a branch and we have no specific suggestion for the middle node.  
\begin{figure}[h!]
				\centering
				\includegraphics[width=1\columnwidth]{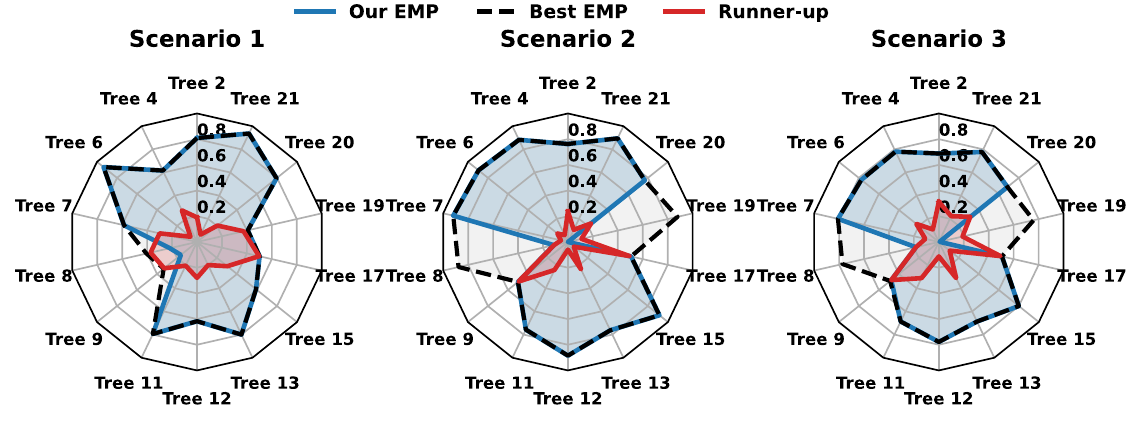}
				\caption{Percentage of the winning EMPs for the trees with suggestions from our insights.}
				\label{fig:radar}
\end{figure}
 
Considering all tree topologies and scenarios depicted in Figure \ref{fig:radar} (except for trees \#8, \#9 and \#19), we found that in average our EMP is the best in $75.45\%$ of cases and 
is the champion in all topologies and scenarios for which there is a specific suggestion.  
Equally important is to realize how relevant this is. 
The ratio between the precision of best EMP and the worst EMP is in average $1.14\cdot {10}^7$, and can be as large as $6.76 \cdot 10^{32}$.

\section{Conclusions}
\label{sec:conclusions}

In this work we have shown how to obtain the information matrix for any dynamic network through the concept of partial information matrix.
In addition, we have shown, through the analysis of the cross networks, that a principle to select EMPs that yield more accurate estimates for tree is to measure nodes with multiple excited in-neighbors and excite nodes with multiple measured out-neighbors. 
Moreover, we have provided a condition that extends the one from branches in which the accuracy of a particular module does not depend on the rest of the network, but only on the SNR applied at the adjacent nodes. 
This is of particular importance for the estimation of this module, once one can obtain more accurate estimates if the parameter has small magnitude.
In addition, we have provided a numerical analysis for 5 node trees under which we validated the principles obtained for the crosses. 
The relevance of these principles are illustrated by the large difference in terms of accuracy observed in the selection of the best EMP when compared to the runner up EMP.

\bibliography{treebib}

\appendix

\section{5 node trees}

Figure \ref{fig:alltrees5nodes} depicts the 21 trees used in the numerical analysis.

\begin{figure*}
	\includegraphics[width=\textwidth, height=\textheight]{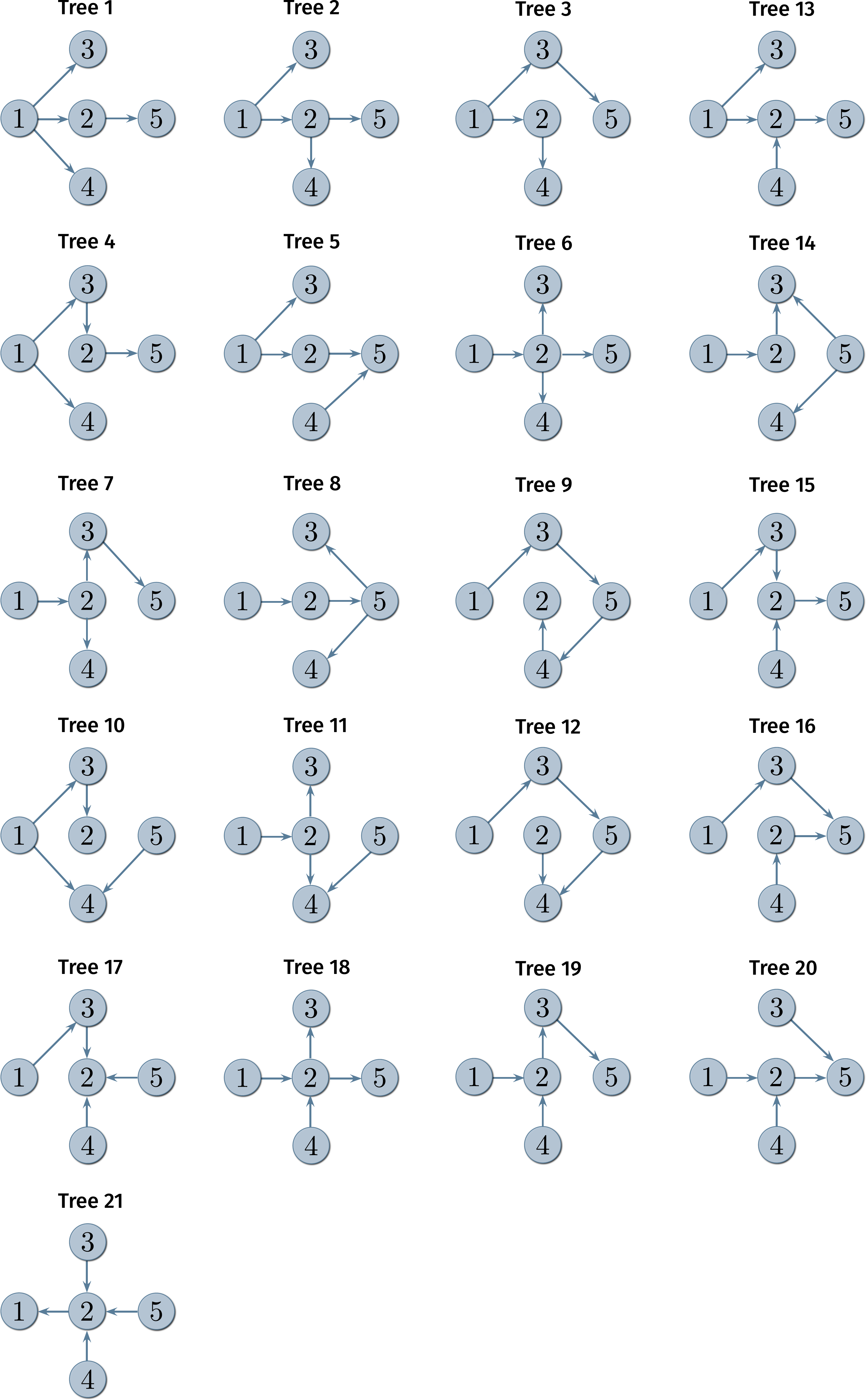}
  \caption{All different 5 node trees.}
	\label{fig:alltrees5nodes}
\end{figure*}

\end{document}